\documentclass[a4paper]{article}
\usepackage{amsmath, amsthm, amssymb}
\usepackage[usenames]{color}
\usepackage{tikz}
\usepackage{enumerate}
\usepackage{bm}

\newtheorem{theorem}{Theorem}
\newtheorem{definition}{Definition}
\newtheorem{lemma}{Lemma}
\newtheorem{example}{Example}

\newcommand{\NP}{\mbox{\bf NP}}

\def\A{\Gamma}
\def\Ast{\mathcal{A}}

\def\Xst{\mathcal{X}}
\def\B{\mathbb{B}}

\def\DG{\mathbb{D}}
\def\Fan{\mathbb{F}}
\def\G{\mathbb{G}}
\def\I{\mathbb{I}}
\def\Q{\mathbb{Q}}
\def\H{\mathbb{H}}
\def\G{\mathbb{G}}
\def\P{\mathbb{P}}
\def\Z{\mathbb{Z}}
\def\F{\Gamma}
\def\I{{\cal I}}
\def\zd{,\ldots,}

\newcommand{\Pol}{\mathrm{Pol}}
\newcommand{\wPol}{\mathrm{wPol}}
\newcommand{\CSP}{\mathrm{CSP}}
\newcommand{\VCSP}{\mathrm{VCSP}}
\newcommand{\MCHom}{\mathrm{MinCostHom}}

\title{A Reduction from Valued CSP to Min Cost Homomorphism Problem for Digraphs\footnote{This work was partially supported by the UK EPSRC grants EP/H000666/1 and EP/J000078/1.}}
\author{Robert Powell and Andrei Krokhin \\\\
Durham University\\
School of Engineering and Computing Sciences\\
DH1 3LE, Durham, UK\\
\texttt{firstname.lastname@durham.ac.uk}}

\begin{document}

\maketitle

\begin{abstract}
In a valued constraint satisfaction problem (VCSP), the goal is to find an assignment of labels to variables that minimizes a given sum of functions. Each function in the sum depends on a subset of variables, takes values which are rational numbers or infinity, and is chosen from a fixed finite set of functions called a constraint language. The case when all functions take only values 0 and infinity is known as the constraint satisfaction problem (CSP). It is known that any CSP with fixed constraint language is polynomial-time equivalent to one where the constraint language contains a single binary relation (i.e. a digraph). A recent proof of this by Bulin {\em et al.} gives such a reduction that preserves most of the algebraic properties of the constraint language that are known to characterize the complexity of the corresponding CSP. We adapt this proof to the more general setting of VCSP to show that each VCSP with a fixed finite (valued) constraint language is equivalent to one where the constraint language consists of one $\{0,\infty\}$-valued binary function (i.e. a digraph) and one finite-valued unary function, the latter problem known as the (extended) Minimum Cost Homomorphism Problem for digraphs. We also show that our reduction preserves some important algebraic properties of the (valued) constraint language.
 \end{abstract}

\section{Introduction}

The constraint satisfaction problem (CSP) is a well studied framework that can express a number of combinatorial problems including propositional satisfiability, graph colouring, and systems of equations. An instance of CSP consists of a set of variables, a set of labels those variables can take, and a set of constraints specifying combinations of labels that certain subsets of the variables can take, with the goal of finding an assignment of labels to variables that satisfies the constraints. It is well known that CSP is in general $\NP$-complete, but by restricting the problem to a fixed set of constraint types, generally referred to as constraint languages~\cite{Bula05:algdichotomy,FederVardi98:cs}, one can obtain polynomial time solvable subproblems. Classifying the complexity of CSPs with a fixed constraint language has been a major area of research since Schaefer's pioneering dichotomy result~\cite{Sch78:comp},
see, e.g.~\cite{Barto09:boundedWidth,Bula06:threeElement,Bula05:algdichotomy,Bulin13:reduc,Bulin14:reduc,FederVardi98:cs}.

There are several natural optimisation versions of CSP. We will consider two of them: Minimum Cost Homomorphism (MinCostHom) and Valued CSP (VCSP).
In MinCostHom, an instance of CSP comes with additional cost functions specifying the cost of assigning each label to each variable; the goal is then to decide whether a satisfying assignment exists and if so find one with minimal total cost. Note that one may naturally
use a restricted set of available cost functions, considering it as part of a constraint language. Complexity classification results for MinCostHom can be found, e.g., in~\cite{Gutin:mincostdichotomy,Hell12:sidma,Tak10:MinCostHom,Tak12:extMinCostHom,Uppman13:icalp,Uppman14:mincosthom}.
The VCSP is the most general optimisation version of CSP, where each constraint, instead of specifying allowed combinations of labels for its variables, assigns each combination of labels a cost, which is a rational number or infinity (the latter indicates disallowed combinations).
The goal is then to find an assignment with minimal total cost. Naturally, constraint languages then consist of cost functions instead of relations.
There has been significant interest in classifying the complexity of VCSPs, see survey~\cite{Jeavons14:survey}, also~\cite{Cohen12:weightedPol,Cohen06:soft,Huber13:skew,Och12:rigidCore,Thap12:linearProg,Thap12:dichotomy} for recent results.

Note that the CSP deals entirely with the feasibility issue (can all constraints be satisfied?), MinCostHom adds a limited
optimisation aspect, as cost functions are applied only to individual variables, while VCSP fully incorporates both feasibility and optimisation issues.
While the full complexity classification even for CSP is open, and in fact is a major open problem~\cite{Bula05:algdichotomy,FederVardi98:cs},
it is interesting to find out how the difficulties of classifying these three frameworks relate to each other.
We show, somewhat surprisingly, that classification of VCSP reduces to classification of MinCostHom, even to the case
when MinCostHom has constraint language consisting only of a digraph and one (unary) cost function.

It is well known that the CSP can be cast as a homomorphism problem for relational structures~\cite{FederVardi98:cs}, the
special case being the (much studied) homomorphism problem for (di)graphs~\cite{Hell04:book}. The CSP with a fixed constraint language
then becomes the problem of deciding the existence of a homomorphism into a fixed relational structure. It was shown in~\cite{FederVardi98:cs} that
every CSP (with a fixed target structure) is polynomial-time equivalent to a digraph homomorphism problem (with a fixed target digraph)
and recently Bulin et al.~\cite{Bulin13:reduc,Bulin14:reduc} developed a variant of this reduction that maintained a number of useful algebraic properties that are crucial in the (extremely successful) algebraic approach to the (V)CSP~\cite{Barto09:boundedWidth,Bula05:algdichotomy,Cohen12:weightedPol,Jeavons14:survey,Thap12:linearProg}. It was explicitly asked in~\cite{Bulin14:reduc} whether their technique can be extended to other constraint problems such as VCSP, and we show that indeed it can. In fact, our proofs reuse many parts from~\cite{Bulin13:reduc,Bulin14:reduc}.

We note that after the results of this paper were announced in~\cite{Jeavons14:survey}, it was shown in~\cite{Kolmogorov15:complexity}
that complexity classification for VCSPs reduces to such a classification for CSPs.

\section{Preliminaries}


\subsection{Valued Constraint Satisfaction and Homomorphism Problems}
Let $D$ be a fixed finite set. Let $\Q_+$ ($\overline{\Q}_+$) denote the set of non-negative rational numbers (with positive infinity, respectively). Let $F^{(m)}_D$ be the set of all $m$-ary functions from $D^m$ to $\overline{\Q}_+$, and then $F_D = \bigcup_{m=1}^\infty F^{(m)}_D$. We will often call functions in $F_D$ {\em cost functions}.
 For the remainder of this paper we only consider such finite constraint languages.

\begin{definition}\label{VCSPInstance}
Let $V=\{x_1\zd x_n\}$ be a set of variables. A {\em valued constraint} over $V$ is an expression,
$\phi(\mathbf{x})$, where $\mathbf{x}\in V^m$ and $\phi\in F^{(m)}_D$.

An instance $\I$ of VCSP is a function $f_{\I}(x_1\zd x_n) = \sum_{i=1}^q{w_i\cdot \phi_i(\mathbf{x}_i)}$ where each $\phi_i(\mathbf{x}_i)$ is a valued constraint over $V_\I=\{x_1\zd x_n\}$ and the $w_i$'s are non-negative rational weights. The goal is to find a mapping $h:V\rightarrow D$ that minimises $f_{\I}$.
\end{definition}

\begin{definition}\label{VCSPGamma}
A {\em valued constraint language}, $\F$, over $D$ is a {\em finite} subset of $F_D$.
For a valued constraint language $\Gamma$, $\VCSP(\Gamma)$ is the class of all instances of VCSP where the cost functions of all the valued constraints are in $\Gamma$.
\end{definition}

The VCSP framework is sufficiently flexible to describe a number of well known problems, as highlighted in the following examples.

\begin{example}\label{MaxCutExample}
Consider the well known $\NP$-hard problem MAX CUT, where given an edge-weighted graph the aim is to partition the vertices into two sets and maximise the total weight of the edges with endpoints in different sets. It is easy to see that this problem can be equivalently expressed as a VCSP. Let $\phi_{MC}:\{0,1\}^2 \to \Q$ be such that $\phi_{MC}(0,1)=\phi_{MC}(1,0)<\phi_{MC}(0,0)=\phi_{MC}(1,1)$. Let $\Gamma_{MC}=\{\phi_{mc}\}$, then $\VCSP(\Gamma_{MC})$ is equivalent to MAX CUT, thus $\VCSP(\Gamma_{MC})$ is also $\NP$-hard.
\end{example}

\begin{example}\label{CSPasVCSP}
The standard CSP can be expressed as a VCSP, where all cost functions are $\{0,\infty\}$-valued, representing allowed and disallowed tuples, respectively. Valued constraints using such $\{0,\infty\}$-valued cost functions are often referred to as {\em crisp} constraints.
\end{example}

\begin{example}\label{submodularEx}
Let $(D,\vee,\wedge)$ be an arbitrary lattice. A function $\phi:D^n\rightarrow\Q_+$ is called submodular if it satisfies the inequality
\[\phi(\textbf{a}\vee\textbf{b})+\phi(\textbf{a}\wedge\textbf{b})\leq \phi(\textbf{a})+\phi(\textbf{b})\; \mbox{ for all } \textbf{a},\textbf{b}\in D^n.\]
If a constraint language $\Gamma$ consists of submodular functions then $\VCSP(\Gamma)$ is tractable~\cite{Thap12:linearProg}.
\end{example}

More examples of (hard and tractable) problems that can be expressed as VCSPs can be found in recent survey~\cite{Jeavons14:survey}.


We now explain how CSP and VCSP can be cast as homomorphism problems. 

\begin{definition}
Let $\tau$ be a {\em relational signature}, that is a set of relational symbols $R$ each with an associated arity $ar(R)$. A {\em (relational) $\tau$-structure} $\Ast$ consists of a finite domain $D$ together with a relation $R^{\Ast}$ on $D$ of arity $ar(R)$ for each $R\in\tau$.
If $\Xst$ and $\Ast$ are $\tau$-structures with domains $X$ and $D$, respectively, then a {\em homomorphism} from $\Xst$ to $\Ast$ is a function $h:X\rightarrow D$ such that, for each $R\in\tau$ and each tuple $\textbf{x}\in R^{\Xst}$, we have $h(\textbf{x})\in R^{\Ast}$
where $h$ is applied component-wise. In this case, we write $h:\Xst\rightarrow \Ast$.
\end{definition}

\begin{definition}
Let $\Ast$ be a finite relational $\tau$-structure. Then $\CSP(\Ast)$ is the following decision problem: Given a $\tau$-structure $\Xst$, is there a homomorphism from $\Xst$ to $\Ast$?
\end{definition}

\begin{example}
The digraph homomorphism problem for a fixed digraph $\H$ asks whether an input digraph $\G$ admits a homomorphism to $\H$, that is whether there is a mapping $h:V^{\G}\rightarrow V^{\H}$ such that if $(u,v)\in E^{\G}$ then $(h(u),h(v))\in E^{\H}$.
This problem is also known as the $\H$-colouring problem~\cite{Hell04:book}. If $\H$ is the complete graph on $k$ vertices then
this problem is the well-known $k$-colouring problem.
\end{example}

It is known that restricting CSP from general structures to digraphs does not reduce the difficulty of classifying the complexity.

\begin{theorem}[\cite{Bulin13:reduc,Bulin14:reduc,FederVardi98:cs}]
For every structure $\Ast$, there is a digraph $\H$ such that $\CSP(\Ast)$ and $\CSP(\H)$ are polynomial-time equivalent.
\end{theorem}

\begin{definition}\label{def:MCHom}
In the $\MCHom(\Ast)$ problem, one is given an input $\tau$-structure $\Xst$ and, in addition, for each $x\in X$, a unary cost function $u_x:D\rightarrow \Q_+$ specifying the cost of mapping $x$ to each individual element in $D$. The goal is to decide whether there is a homomorphism $h$ from $\Xst$ to $\Ast$ and if so find one of minimal total cost $\sum_{x\in X}{u_x(h(x))}$. For a set $\Delta$ of unary cost functions,
let $\MCHom(\Ast,\Delta)$ denote the subproblem of $\MCHom(\Ast)$ where all unary functions in instances are of the form $w\cdot u$ where
$w\in\Q_+$ and $u\in \Delta$. If $\Delta=\{u\}$, we write simply $\MCHom(\Ast,u)$.
\end{definition}

The problem $\MCHom(\Ast)$ was studied in a series of papers, and complete complexity classifications were given
in~\cite{Gutin:mincostdichotomy} for undirected graphs, in~\cite{Hell12:sidma} for digraphs, and in~\cite{Tak10:MinCostHom} for general
structures. Partial complexity classifications for the problem $\MCHom(\Ast,\Delta)$ were obtained in~\cite{Tak12:extMinCostHom,Uppman13:icalp,Uppman14:mincosthom}.
One can see that MinCostHom is an intermediate problem between CSP and VCSP, as there is an optimisation aspect, but it is limited in the sense that it is controlled by separate unary functions, without explicit interactions of variables.

We will now define VCSP as a homomorphism problem, following~\cite{Thap12:linearProg} (see also~\cite{Cohen12:weightedPol}). This will allow us to easily reuse many results from~\cite{Bulin13:reduc,Bulin14:reduc}.


\begin{definition}\label{homomorphismVCSP}
A {\em weighted relation} $\rho$ of arity $k$ on a set is a function from some $k$-ary relation $R$ on this set to $\Q_+$.
 A {\em weighted $\tau$-structure} $w\Ast$ is $\tau$-structure such that each relation $\rho^{w\Ast}$ in $w\Ast$ is weighted, i.e. $\rho^{w\Ast}:R^{\Ast}\rightarrow \Q_+$. By ignoring the weight functions, one can turn a weighted $\tau$-structure $w\Ast$ into an ordinary, unweighted, $\tau$-structure $\Ast$.

An instance of $\VCSP(w\Ast)$ is given by a weighted $\tau$-structure $w\Xst$. A feasible solution to this instance is a homomorphism $h$ from $\Xst$ to $\Ast$, and its cost is given by
\[cost(h)=\sum_{R\in\tau,\mathbf{x}\in R^{\Xst}}{\rho^{w\Xst}(\textbf{x})\cdot \rho^{w\Ast}(h(\textbf{x}))}.\]
The goal is to decide if such a homomorphism exists, and if so find one with minimal cost.
\end{definition}

\begin{definition}\label{Feas}
For an  $m$-ary cost function $\phi:D^m\rightarrow\overline{\Q}_+$, we define the \textit{feasibility relation}, $Feas(\phi)$, of $\phi$ as follows:
$(x_1,x_2,...,x_m)\in Feas(\phi)\Leftrightarrow \phi(x_1,x_2,...,x_m)<\infty.$
\end{definition}

There is an obvious correspondence between valued constraint languages and weighted structures: any constraint language $\Gamma$ can be turned into
a weighted structure $w\Ast$ as follows. Turn each function $\phi\in\Gamma$ into a weighted relation $\rho:Feas(\phi)\rightarrow \Q_+$, simply by ignoring the tuples with infinite cost, then introduce a signature $\tau$ containing a symbol $R_\phi$ of arity $k$ for each function $\phi\in \Gamma$
of arity $k$. Clearly, one obtains a weighted $\tau$-structure. One can also reverse this procedure to convert a weighted structure into a valued constraint language.

The correspondence between $\VCSP(w\Ast)$ and $\VCSP(\Gamma)$ can be seen as follows. If $w\Xst$ is an instance of $\VCSP(w\Ast)$,
then one can view the domain of $w\Xst$ as the set of variables in an instance $\I$ of $\VCSP(\Gamma)$, and each tuple $\mathbf{x}\in R^{\Xst}$
gives rise to a valued constraint $\phi(\mathbf{x})$ with weight $\rho^{w\Xst}(\mathbf{x})$ where $\phi\in F_D$ is the function obtained
by extending $\rho^{w\Ast}$ with infinite values. Then the homomorphisms from $\Xst$ to $\Ast$ are precisely the solutions to $\I$ of finite cost, and the correspondence preserves the costs. Thus, $\VCSP(w\Ast)$ and $\VCSP(\Gamma)$ are effectively the same problem.

Note that if all functions in $\Gamma$ are $\{0,\infty\}$-valued, i.e. $\VCSP(\Gamma)$ is in fact a CSP, then the weighted structure $w\Ast$ obtained from $\Gamma$ as described above will be 0-weighted, i.e. effectively unweighted, and $\VCSP(w\Ast)$ is the same problem as $\CSP(\Ast)$.
It also clear that if all functions in $\Gamma$ are $\{0,\infty\}$-valued or unary finite-valued, then $\VCSP(\Gamma)$ is $\MCHom(\Ast,\Delta)$
for the obvious choice of $\Ast$ and $\Delta$.

\subsection{Algebra}

\begin{definition}\label{polymorphism}
Let $\rho:R\rightarrow \Q_+$ be a weighted relation on $D$ with underlying relation $R$. We say that an operation $f:D^k\rightarrow D$ is a {\em polymorphism} of $\rho$ (and of $R$) if, for any $\bm{x_1},\bm{x_2},...,\bm{x_k}\in R$ we have $f(\bm{x_1},\bm{x_2},...,\bm{x_k})\in R$, where $f$ is applied component-wise.

 For a weighted structure $w\Ast$, we denote by $\Pol(w\Ast)$ the set of all operations which are polymorphisms of each $\rho$ in $w\Ast$, and by $\Pol^{(k)}(w\Ast)$ the set of $k$-ary operations in  Pol$(w\Ast)$.
\end{definition}

Polymorphisms have played a major role in the classifications of complexity for ordinary CSP~\cite{Barto09:boundedWidth,Bula06:threeElement,Bula05:algdichotomy}, but for VCSP we need the more general notion of weighted polymorphisms~\cite{Cohen12:weightedPol}.

\begin{definition}\label{weightedPoly}
Let $\rho:R\rightarrow \Q_+$ be a weighted relation and let $C\subseteq \Pol^{(k)}(\rho)$. A function $\omega:C\rightarrow\Q$ is a $k$-ary {\em weighted polymorphism} of $\rho$ if it satisfies the following conditions:
\begin{itemize}
	\item $\sum_{f\in C} \omega(f) = 0$;
	\item if $\omega(f)<0$ then $f$ is a projection, i.e., for some $1\le i\le k$, $f(x_1,\ldots,x_k)=x_i$;
	\item for any $\bm{x_1},\bm{x_2},...,\bm{x_k}\in R$
	\[\sum\limits_{f\in C}\omega(f)\cdot \rho(f(\bm{x_1},\bm{x_2},...,\bm{x_k}))\leq 0\]
\end{itemize}
Let $\wPol(\phi)$ denote the set of all weighted polymorphisms of $\rho$.
\end{definition}

\begin{example}\label{submodAsWeighted}
We can rewrite the submodularity condition from Example~\ref{submodularEx} as a binary weighted polymorphism. Let the functions $f_1,f_2,f_3,f_4$ be the operations $\vee$, $\wedge$, $\mathrm{Pr}_1$,$\mathrm{Pr}_2$ where the last two operations are projections on the first and second coordinate, respectively.
Consider the weighted operation that assigns these operations the respective weights $1, 1 ,-1 , -1$. It is easy to check that any function $\phi\in F_D$ has this weighted polymorphism if and only if it is submodular.
\end{example}

For a weighted structure $w\Ast$, let $\wPol(w\Ast)=\bigcap_{\rho\in w\Ast}{\wPol(\rho)}$.
Weighted polymorphisms are powerful algebraic tools in the study of VCSP, and it has been shown that the complexity of valued constraint languages can be characterised by its weighted polymorphisms~\cite{Cohen12:weightedPol}.

\begin{theorem}[\cite{Cohen12:weightedPol}]
If $w\Ast_1$ and $w\Ast_2$ are weighted structures on $D$ such that $\wPol(w\Ast_1)\subseteq \wPol(w\Ast_2)$
then $\VCSP(w\Ast_2)$ is polynomial-time reducible to $\VCSP(w\Ast_1)$.
\end{theorem}

Weighted polymorphisms (and their special cases) play a key role in complexity classifications for VCSP~\cite{Cohen12:weightedPol,Cohen06:soft,Jeavons14:survey,Thap12:dichotomy,Uppman13:icalp,Uppman14:mincosthom}.

We now introduce the notion of a core valued constraint language. Intuitively, a valued constraint language is not a core if there is some element of its domain, $a\in D$, such that any instance has an optimal solution that does not use $a$. We can simply remove $a$ from $D$, reducing the problem to one on a smaller domain. 
Formally, cores are defined as  follows.

\begin{definition}\label{rigidCore}
A weighted structure $w\Ast$ is a {\em core} if all its unary polymorphisms are bijections, and it is a {\em rigid core} if the identity mapping is the only unary polymorphism of $w\Ast$.
\end{definition}

\begin{lemma}[\cite{Och12:rigidCore}]
For every weighted structure $w\Ast$, there is a weighted structure $w\Ast'$  such that $w\Ast'$ is a rigid core and $\VCSP(w\Ast')$ is polynomial-time equivalent to $\VCSP(w\Ast)$.
\end{lemma}


\section{Results}

In this section we formally state and prove the main results of this paper. 

\begin{theorem}\label{thm:main}
Let $w\Ast$ be a weighted structure that is a rigid core. There is a balanced digraph $\DG$ which is a rigid core and a finite-valued function $u$ such that  problems $\VCSP(w\Ast)$ and $\MCHom(\DG,u)$ are polynomial-time equivalent.
\end{theorem}

\begin{proof}
We can assume without loss of generality that $w\Ast$ contains only one weighted relation, say of arity $n$.
If it contains more, say, $\rho_1,\ldots,\rho_t$ where, for $1\le j\le t$, the arity of $\rho_j$ is $n_j$, than the standard trick is to take the direct product $\rho$ of these relations. Specifically, if $\rho_j:R_j\rightarrow \Q_+$ then $\rho:R_1\times \ldots \times R_t \rightarrow \Q_+$
is such that $\rho(\mathbf{a}_1,\ldots,\mathbf{a}_t)=\rho_1(\mathbf{a}_1)+\ldots+\rho_t(\mathbf{a}_t)$.
It is well known and not hard to see that replacing $\rho_1,\ldots,\rho_t$ with $\rho$ does not change the complexity of $\VCSP(w\Ast)$.

The reduction from $\VCSP(w\Ast)$ to $\MCHom(\DG,u)$ follows from  Lemmas~\ref{polyreduclem} and~\ref{expressiblelem2}, and the reduction in the opposite direction is shown in Lemma~\ref{reduceDtoAunary}.
\end{proof}


\subsection{Constructing $(\DG,u)$}

Firstly we introduce some simple definitions related to digraphs that we will need in the construction of the main theorem of this paper.


\begin{definition}
A digraph $\P$ is an oriented path if it consists of a sequence of vertices $v_0,v_1,...,v_k$ such that precisely one of $(v_{i-1},v_i),(v_i,v_{i-1})$ is an edge, for each $i=1,...,k$. We denote the initial vertex $v_0$ by $\iota\P$ and the terminal vertex $v_k$ by $\tau\P$.
\end{definition}

\begin{definition}
Given a digraph $\G$, any two vertices $a$ and $b$ in $\G$ are connected if there is an oriented path between them. We write $a\xrightarrow{\P} b$ if there exists a homomorphism $\phi:\P\rightarrow\G$ such that $\phi(\iota\P)=a$ and $\phi(\tau\P)=b$.
\end{definition}

A digraph $\G$ is said to be (weakly) connected if every pair of vertices in it are connected. The length of an oriented cycle is defined as being the absolute value of the difference between edges oriented in one direction and edges oriented in the opposite direction. A connected digraph is {\em balanced} if all of its cycles have zero length~\cite{FederVardi98:cs}. The vertices of a balanced digraph can then be organised into levels, that is for every edge $(a,b)$ in the digraph $\G$, $lvl(b)=lvl(a)+1$. The minimum level of $\G$ is 0, and the top level is the height of $\G$. We will sometimes call level 0 vertices {\em base vertices}.


As in~\cite{Bulin13:reduc,Bulin14:reduc}, we say that a {\em zigzag} is the oriented path $\bullet\to\bullet\leftarrow\bullet\to\bullet$ and a {\em single edge} is the path $\bullet\to\bullet$.

Recall that $n$ is the arity of the single weighted relation $\rho$ in $w\Ast$.
For $S\subseteq \{1,2\zd n\}$ define $\Q_{S,l}$ to be a single edge if $l\in S$, and a zigzag if $l\in \{1,2\zd n\}\setminus S$.
As in~\cite{Bulin13:reduc,Bulin14:reduc} we define the oriented path $\Q_S$ (of height $n+2$) by \[\Q_{S}=\bullet\to\bullet\; \dot{+}\; \Q_{S,1} \;\dot{+}\; \Q_{S,2} \;\dot{+}\;\dotsb\;\dot{+}\; \Q_{S,n} \;\dot{+}\;\bullet\to\bullet\]
where $\dot{+}$ denotes the concatenation of paths.

Now take $R$, i.e. the underlying relation of $\rho$, and the domain $D$ of $w\Ast$, and as a starting point consider the digraph with vertices $D\cup R$ and edges $D\times R$. Then replace every edge $(d,\mathbf{a})\in D\times R$ with the path $\Q_{\{i:d=a_i\}}$, and this is the digraph $\DG$. \\
To complete the construction,  we define a unary function $u$. Let $V^{\DG}$ be the vertices of $\DG$, and note that $R\subseteq V^{\DG}$.  Let $u$ to be a unary function from $V^{\DG}$ to $\Q_+$ such that
\[u(v)=\left\{\begin{array}{l l}\rho(v) & \quad\text{if }v\in R\\0 & \quad \text{otherwise.}\end{array}\right.\]
The digraph $\DG$ is identical to the digraph defined in~\cite{Bulin13:reduc,Bulin14:reduc} and due to our definition of $R$ the number of vertices in $\DG$ remains as $(3n+1)|R||D|+(1-2n)|R|+|D|$ and the number of edges as $(3n+2)|R||D|-2n|R|$ as proven in~\cite{Bulin13:reduc,Bulin14:reduc}. Also as noted in~\cite{Bulin13:reduc,Bulin14:reduc} this construction can be performed in polynomial time.

\begin{example}\label{ex1}
Consider the weighted structure $w\Ast$ over the domain $D=\{0,1\}$ with the single weighted relation \[\rho(x,y)=\left\{\begin{array}{l l}2 & \quad\text{if }(x,y)=(0,1)\\1 & \quad \text{if }(x,y)=(1,0)\\\text{undefined} & \quad \text{otherwise.}\end{array}\right.\]
The digraph $\DG$ constructed from $\rho$ is shown in Figure 1. The unary function built from $\rho$ is \[u(v)=\left\{\begin{array}{l l}2 & \quad\text{if }v=(0,1)\\1 & \quad \text{if }v=(1,0)\\0 & \quad \text{otherwise}\end{array}\right.\] for every vertex $v\in V^{\DG}$.
\end{example}

\begin{figure}\label{digraphEx1}
\begin{center}
\begin{tikzpicture}

\fill (canvas cs:x=2cm,y=5cm) circle (2pt);
\fill (canvas cs:x=11cm,y=5cm) circle (2pt);
\fill (canvas cs:x=5cm,y=1cm) circle (2pt);
\fill (canvas cs:x=8cm,y=1cm) circle (2pt);

\draw (canvas cs:x=1cm,y=4cm) circle (2pt);
\draw (canvas cs:x=3cm,y=4cm) circle (2pt);
\draw (canvas cs:x=5cm,y=4cm) circle (2pt);
\draw (canvas cs:x=8cm,y=4cm) circle (2pt);
\draw (canvas cs:x=10cm,y=4cm) circle (2pt);
\draw (canvas cs:x=12cm,y=4cm) circle (2pt);

\draw (canvas cs:x=1cm,y=3cm) circle (2pt);
\draw (canvas cs:x=3cm,y=3cm) circle (2pt);
\draw (canvas cs:x=4cm,y=3cm) circle (2pt);
\draw (canvas cs:x=6cm,y=3cm) circle (2pt);
\draw (canvas cs:x=7cm,y=3cm) circle (2pt);
\draw (canvas cs:x=9cm,y=3cm) circle (2pt);
\draw (canvas cs:x=10cm,y=3cm) circle (2pt);
\draw (canvas cs:x=12cm,y=3cm) circle (2pt);

\draw (canvas cs:x=2cm,y=2cm) circle (2pt);
\draw (canvas cs:x=4cm,y=2cm) circle (2pt);
\draw (canvas cs:x=6cm,y=2cm) circle (2pt);
\draw (canvas cs:x=7cm,y=2cm) circle (2pt);
\draw (canvas cs:x=9cm,y=2cm) circle (2pt);
\draw (canvas cs:x=11cm,y=2cm) circle (2pt);

\node [below] at (5,0.95) {$0$};
\node [below] at (8,0.95) {$1$};
\node [above] at (2,5) {$(1,0)$};
\node [above] at (11,5) {$(0,1)$};

\draw [->,thick](1.1,4.1) -- (1.9,4.9);
\draw [->,thick](1,3.15) -- (1,3.85);
\draw [->,thick](1.9,2.1) -- (1.1,2.9);
\draw [->,thick](2.1,2.1) -- (2.9,2.9);
\draw [->,thick](3.9,2.1) -- (3.1,2.9);
\draw [->,thick](4.9,1.1) -- (4.1,1.9);

\draw [->,thick](10.1,4.1) -- (10.9,4.9);
\draw [->,thick](9.1,3.1) -- (9.9,3.9);
\draw [->,thick](8.9,3.1) -- (8.1,3.9);
\draw [->,thick](7.1,3.1) -- (7.9,3.9);
\draw [->,thick](6.1,2.1) -- (6.9,2.9);
\draw [->,thick](5.1,1.1) -- (5.9,1.9);

\draw [->,thick](2.9,4.1) -- (2.1,4.9);
\draw [->,thick](3.9,3.1) -- (3.1,3.9);
\draw [->,thick](4.1,3.1) -- (4.9,3.9);
\draw [->,thick](5.9,3.1) -- (5.1,3.9);
\draw [->,thick](6.9,2.1) -- (6.1,2.9);
\draw [->,thick](7.9,1.1) -- (7.1,1.9);

\draw [->,thick](11.9,4.1) -- (11.1,4.9);
\draw [->,thick](12,3.15) -- (12,3.85);
\draw [->,thick](11.1,2.1) -- (11.9,2.9);
\draw [->,thick](10.9,2.1) -- (10.1,2.9);
\draw [->,thick](9.1,2.1) -- (9.9,2.9);
\draw [->,thick](8.1,1.1) -- (8.9,1.9);

\end{tikzpicture}

\caption{The digraph $\DG$ built from the weighted structure $w\Ast$.}
\end{center}
\end{figure}

Lemma 4.1 of~\cite{Bulin14:reduc} states that the unary polymorphisms of relation $R$ and of digraph $\DG$ are in one-to-one correspondence.
Hence, we immediately get the following.

\begin{lemma} $w\Ast$ is a rigid core if and only $\DG$ is a rigid core.
\end{lemma}

\subsection{Reduction from $\VCSP(w\Ast)$ to $\MCHom(\DG,u)$}

Given an instance of $\VCSP(w\Ast_0)$, where $w\Ast_0$ is some weighted structure, the variables in any constraint in that instance are explicitly constrained. It is also possible that any subset of variables in that instance are implicitly constrained due to combinations of constraints. The weighted relation describing this implicit constraint may not belong to $w\Ast_0$, but is said to be expressible by $w\Ast_0$.

\begin{definition}
Given an instance $w\Xst$ of $\VCSP(w\Ast_0)$ with domain $X=\{x_1,...,x_n\}$, and a tuple of distinct elements $W=(x_1 \zd x_t)$, where $t\leq n$, define weighted relation $\rho_{w\Xst}^{W}$ as follows:
\begin{equation}\label{expressEq}
\rho_{w\Xst}^{W}(a_1\zd a_t)=\min_{h:\Xst\rightarrow \Ast,\ h(x_1,\ldots,x_t)=(a_1,\ldots,a_t)}{cost(h)}.
\end{equation}
Note that if, for some $(a_1\zd a_t)$, there is no homomorphism $h:\Xst\rightarrow \Ast$ such that $h(x_1,\ldots,x_t)=(a_1,\ldots,a_t)$ then
$\rho_{w\Xst}^{W}(a_1\zd a_t)$ is undefined.
A weighted relation $\rho$ is {\em expressible} by $w\Ast_0$ if there is an instance $w\Xst$ of $\VCSP(w\Ast_0)$ with domain $X$, and $W\subseteq X$, such that $\rho=\rho_{w\Xst}^{W}$.
\end{definition}

\begin{lemma}[\cite{Cohen06:soft}]\label{polyreduclem}
Let $w\Ast$ and $w\Ast_0$ be weighted structures. If every weighted relation $\rho$ in $w\Ast$ is expressible by $w\Ast_0$, then $\VCSP(w\Ast)$ is polynomial-time reducible to $\VCSP(w\Ast_0)$.
\end{lemma}

Let $w\Ast_0$ be the weighted relational structure containing only weighted relations $E^\DG$ and $u$, where $E^\DG$ simply assigns 0 to every edge of $\DG$.

\begin{lemma}\label{expressiblelem2}
$w\Ast$ is expressible by $w\Ast_0$.

\end{lemma}

\begin{proof}
All we need to see is that the domain $D$ and weighted relation $\rho$ of $w\Ast$ can be expressed by $w\Ast_0$, in the sense of Equation~(\ref{expressEq}).
Consider the weighted structure $w\Xst$ that contains the binary 0-weighted relation corresponding to the oriented path $\Q_{\emptyset}$ (defined above) and the empty unary relation. Let $x$ be the initial vertex of $\Q_{\emptyset}$ and let $W=\{x\}$. Then $D$, as a unary $0$-weighted relation, can be defined
as $\rho_{w\Xst}^{W}$. Indeed, the homomorphisms from $\Xst$ to $\Ast$ are simply homomorphisms from $\Q_{\emptyset}$ to $\DG$, and it is clear that
the images of $x$ under such homomorphisms are exactly the base elements of $\DG$, i.e., precisely the elements of $D$.

To express $\rho$, consider the weighted structure $w\Xst$ whose binary $0$-weighted relation corresponds to the digraph obtained by identifying the terminal vertices of $n$ directed paths $\Q_{\{1\}},\ldots,\Q_{\{n\}}$, and whose unary relation contains a single element $y$, the common vertex of the paths $\Q_{\{1\}},\ldots,\Q_{\{n\}}$, with weight 1. Let $W=\{x_1,\ldots,x_n\}$ where $x_i$ is the initial element of $\Q_{\{i\}}$.
Again, it is not hard to see from the definitions of $\DG$ and $u$ that this $\rho_{w\Xst}^{W}$ is precisely the required weighted relation $\rho$.
\end{proof}

\subsection{Reduction from $\MCHom(\DG,u)$ to $\VCSP(w\Ast)$}


Let $w\Ast'$ be the structure obtained from $w\Ast$ by adding the weighted relation $\rho_0$, which is obtained from $\rho$ by making every weight 0.
It is shown in~\cite{Cohen12:weightedPol} that $\VCSP(w\Ast)$ and $\VCSP(w\Ast')$ are polynomial-time equivalent.

\begin{lemma}\label{reduceDtoAunary}
$\MCHom(\DG,u)$ reduces to $\VCSP(w\Ast')$ in polynomial time.
\end{lemma}

\begin{proof}
Let $(\G,W)$ be an instance of $\MCHom(\DG,u)$, where $\G$ is a digraph and $W$ is a unary weighted relation over $V^{\G}$ responsible for the optimisation aspect of the problem. Formally (to match Definition~\ref{def:MCHom}), we can assume that the vertices of $\G$ outside $W$ have function $0\cdot u$ applied to them.

Our reduction is a modification of the reduction in~\cite{Bulin13:reduc,Bulin14:reduc}. Specifically, Stages 1 and 3a of the reduction are exactly the same, but Stages 2 and 3b are modified.
\\\\
Stage 1: Verify $\G$ is balanced and test height.

This initial check is to verify that that $\G$ is balanced and that $\G$ has height not greater than $m$, the height of $\DG$. If either of these conditions fail there can be no homomorphism from $\G$ to $\DG$ and we return the fixed NO instance of $\VCSP(w\Ast'$). It is easy to see that these checks can be performed in polynomial time.

 It is clear that if $\G$ is balanced and has the right height then, under any homomorphism from $\G$ to $\DG$, only
vertices of top level in $\G$ can be mapped to the vertices of top level in $\DG$. Therefore, any vertex in $W$ that is not from the top level
cannot possibly affect the cost of a homomorphism, and so can be safely removed from $W$. From now we will assume that $W$, if non-empty, contains only top level vertices of $\G$.
\\\\
Stage 2: Elimination of short components.

If $\G$ contains a connected component $\H$ of height less than $m$ (a short component) then we can find an optimal solution for this component directly in polynomial time as in~\cite{Bulin13:reduc,Bulin14:reduc}, the proof of which is in Section~\ref{pathSat}. We repeat this procedure for every short component and if there is any component $\H$ for which there is no solution then we return the fixed NO instance of $\VCSP(w\Ast')$. Provided there is an optimal solution to all short components we can ignore these components for the remainder of the reduction.
If $\G$ itself is of height less than $m$ then finding an optimal solution to every connected component completes the reduction and we output some fixed YES instance of $\VCSP(w\Ast')$.
\\\\
Stage 3a: Construction of $\B'$.

This stage is identical to Stage 3A of~\cite{Bulin13:reduc,Bulin14:reduc} and is included here only for completeness. In this stage we build the ``object'' $\B'$, which consists of a list of tuples, some of them with subscripts, and a list of equalities. These tuples contain sets of vertices of $\DG$ and new vertices created during the algorithm. We will also need the following technical detail from~\cite{Bulin13:reduc,Bulin14:reduc}.

Say that a digraph $\H$ is {\em satisfiable in a digraph} $\H'$ if there is a homomorphism from $\H$ to $\H'$.
It is shown in~\cite{Bulin13:reduc,Bulin14:reduc} that, for any connected balanced digraph $\H$, there is a smallest set $S\subseteq\{1,...,n\}$ such that $\H$ is satisfiable in $\Q_S$, and this set can be efficiently found. We denote this set by $\Gamma(\H)$\footnote{Notation $\Gamma(\H)$ is from~\cite{Bulin13:reduc,Bulin14:reduc}, not related to valued constraint languages $\Gamma$.}.

Any new vertices created in the algorithm to construct $\B'$ should be unique, and we remind the reader that the height of $\G$ is $m$ and the arity of the weighted relation $\rho$ in $w\Ast$ is $n$.
Say that the {\em internal components} of $\G$ are the connected components of the induced subgraph of $\G$ obtained by removing all vertices of height 0 and $m$.
The algorithm is as follows:
\\\\
Firstly we create tuples from the top level vertices of $\G$. For each such vertex $e$, there will be one tuple with subscript $e$.\\
For every vertex $e$ in $\G$ of height $m$ and for $i=1$ to $n$, do the following:
\begin{enumerate}
\item Identify all internal components $C$ of $\G$ such that $i\in\Gamma(C)$ and $C$ has an edge to $e$.
\item If there are no such components, add a new vertex $x$ to the $i$th component (vertex set) of the output tuple for $e$.
\item Else, for each such internal component $C$, do
\begin{enumerate}
	\item If $C$ has edges to vertices of level 0 in $\G$, say $b_1\zd b_j$, then add them to the $i$th vertex set of the output tuple for $e$.
\item Otherwise add a new vertex $x$ to the $i$th vertex set of the output tuple for $e$.
\end{enumerate}
\end{enumerate}

\noindent This completes the first part of the algorithm for constructing $\B'$ and now we create tuples (that will not have subscripts) using the base vertices of $\G$.

For every level 0 vertex $b$ and for $i=1$ to $n$ do the following:
\begin{enumerate}
\item Identify all internal components $C$ of $\G$ such that $C$ has an edge to $b$, but no edge to any vertex of height $m$.
\item For each such internal component, if they exist, do
\begin{enumerate}
	\item if $i\in \Gamma(C)$ then add $b$ to the $i$th vertex set of the output tuple for $b$.
	\item if $i\notin \Gamma(C)$ then add a new vertex $x$ in the $i$th vertex set of the output tuple for $b$.
\end{enumerate}
\end{enumerate}

\noindent The algorithm is completed by creating a list of equalities, $L$, showing (some of the) vertices that will have to be mapped to the same vertex by any homomorphism from $\G$ to $\DG$. The rules for creating $L$ are as follows.
\begin{enumerate}[i.]
	\item If there is an internal component with edges to distinct vertices $e$ and $f$ of height $m$ in $\G$, then we write $e=f$.
	\item If there is an internal component with edges to distinct vertices $b$ and $c$ of height 0 in $\G$, then we write $b=c$.
\end{enumerate}

\noindent Stage 3b: Construction of $w\Xst'$.

In a modification to the construction of structure $\B$ in the proof of~\cite{Bulin13:reduc,Bulin14:reduc}, we construct weighted structure
$w\Xst'$ containing two weighted relations $\rho'_w$ and $\rho'_0$. The relation $\rho'_0$ will be the 0-weighted relation identical to the (unique) relation in structure $\B$ from~\cite{Bulin13:reduc,Bulin14:reduc}, whilst $\rho'_w$ is built using $W$, and controls the optimisation aspect of the reduction.

We start with building the equality graph. Its vertices are the base vertices of $\G$ and the new vertices created in stage 3a of the algorithm. The rules to create edges in the equality graph are taken directly from~\cite{Bulin13:reduc,Bulin14:reduc}.
\\\\
The three rules for creating an edge in the equality graph are:
\begin{enumerate}[i.]
	\item Add an edge from vertex $a$ to vertex $b$ if $a$ and $b$ lie in the same vertex set in $\B'$.
	\item Add an edge if $a=b$ is an equality in $L$.
	\item Add an edge if $a$ and $b$ appear in the $i$th vertex set of two tuples with subscripts $e$ and $f$ where $e=f$ is an equality in $L$.
\end{enumerate}

\noindent
Each element in the domain $X'$ of $w\Xst'$ will be the set of vertices of a connected component of the equality graph.
It follows from the construction that all vertices from the same connected component have to be mapped to the same
element under any homomorphism from $\G$ to $\DG$.

To obtain the tuples of $\rho'_0$ we replace the vertex set in every coordinate of every tuple of $\B'$ with the set of vertices of the connected component containing that vertex set in the equality graph, and remove all of the book-keeping subscripts. Each tuple in the weighted relation $\rho'_0$ is assigned weight 0.

The weighted relation $\rho'_w$ will be defined on (some of the) tuples on which $\rho'_0$ is defined.
Assume that $\rho'_0(A_1,\ldots,A_n)$ is defined. Then $\rho'_w(A_1,\ldots,A_n)$ is defined
if and only if there is $e\in W$  such that the tuple $(B_1,\ldots,B_n)\in \B'$ with subscript $e$ satisfies $B_i\subseteq A_i$ for all $1\le i\le n$. In this case, $\rho'_w(A_1,\ldots,A_n)$ is defined to be the sum of values $w(e)$ over all such $e$.

It is stated in~\cite{Bulin13:reduc,Bulin14:reduc} that homomorphisms from $\G$ to $\DG$ are in one-to-one correspondence with homomorphisms from $\Xst'$ to $\Ast'$ which we elaborate here.
Let $h:\G\rightarrow\DG$, we want to show there is a corresponding homomorphism $h':\Xst'\rightarrow\Ast'$. By construction $\Xst'$ consists of tuples whose elements are vertex sets, where each vertex set can contain bottom vertices of $\G$ and any new vertices created in stage 3a of the algorithm above. Consider all vertex sets of all tuples of $\Xst'$. No two of these vertex sets have any vertices in common, unless the vertex sets are identical, as they would have been identified and grouped together by the equality graph in stage 3b. If a vertex set consists of only new vertices then that vertex set appears only once in one tuple of $\Xst'$ (note that we ignore repeated tuples in $\Xst'$). It is noted in~\cite{Bulin13:reduc,Bulin14:reduc} that any tuple in $\Xst'$ that consists of only new vertices can be ignored. We note that such a tuple can be mapped anywhere by a homomorphism $h$, and therefore the homomorphisms $h$ and $h'$ are only in one-to-one correspondence if we ignore such tuples.

The homomorphism $h$ maps bottom vertices of $\G$ to bottom vertices of $\DG$, and by definition the bottom vertices of $\DG$ are the elements of $\Ast'$. If we consider a tuple in $\Xst'$ where each of its elements is a vertex set containing at least one bottom vertex of $\G$, then we know which element this maps to in $\Ast'$ as we know where $h$ mapped that vertex to in $\DG$. Therefore the homomorphism $h$ fully determines how $h'$ acts on vertex sets containing at least one bottom vertex of $\G$.

Now consider the case of a tuple in $\Xst'$ which has at least one element which is a vertex set containing only new vertices. We'll assume this element is in the $i$th position of the tuple. These new vertices were created in stage 3a, either in step 2 or step 3(b) by a specific top vertex of $\G$, or in step 2(b) by a specific base vertex of $\G$, and will therefore fall into exactly one of the following two cases.

First consider the case where a new vertex was created from a specific top vertex $e$ of $\G$. The new vertex was created as there is either no internal component $C$ with an edge to $e$ such that $i\in\Gamma(C)$ (Step 2), or all the internal components $C$ with edges to $e$ and $i\in\Gamma(C)$ have no edges to any base vertices (Step 3(b)). Given a homomorphism $h:\G\rightarrow\DG$ we can identify which top vertex $t$ in $\DG$ the vertex $e$ is mapped to. Therefore we can identify the unique internal component $C_D$ in $\DG$ with $i\in\Gamma(C_D)$ that has $t$ as its top vertex, and the base vertex of $\DG$ to which it is connected. Thus we have identified the base vertex in $\DG$ (i.e. the element in $\Ast'$) which is mapped to by our vertex set containing only new vertices, and therefore determined how $h'$ acts on this vertex set.

Now consider the case where the new vertex was created by a base vertex $b$ of $\G$ (Step 2(b)). Note that we only create a new vertex if there is at least one internal component of $\G$ with an edge to $b$ that has no edge to any top vertex. Furthermore we note that if there is such an internal component and it has $\Gamma(C)=\emptyset$ the tuple we would obtain in $\Xst'$ would consist of only new vertices and should be ignored. Let $C$ be an internal component with an edge to $b$ and $\Gamma(C)\neq\emptyset$. Under a homomorphism $h:\G\rightarrow\DG$ we can identify where the base vertex $b$ is mapped to in $\DG$, and also where the vertices of $C$ are mapped to. Thus we can identify the specific top vertex $t$ of $\DG$ that is at the end of the oriented path that $C$ mapped to. Now we can use the same method as before to identify the base vertex of $\DG$ that our new vertex must map to. That is we identify the unique internal component $C_D$ in $\DG$ with $i\in\Gamma(C_D)$ that has $t$ as its top vertex, and then identify the base vertex of $\DG$ to which it is connected. The homomorphism $h$ maps our new vertex to that base vertex of $\DG$, and therefore we have determined how $h'$ acts on this vertex set.

The above cases cover all possible vertex sets that can appear in tuples of $\Xst'$, and hence we have fully defined $h'$ given $h$. It remains to argue that $h'$ is a homomorphism from $\Xst'$ to $\Ast'$, i.e. for each tuple $\textbf{x}\in\Xst'$, $h'(\textbf{x})\in\Ast'$. By construction the tuple $\textbf{x}\in\Xst'$ corresponds to a top vertex $e$ of $\G$, which is mapped, by the homomorphism $h$, to a top vertex $t=h(e)$ of $\DG$. The vertex $t$ has a corresponding tuple in $\Ast'$, and this is the tuple $h'(\textbf{x})\in\Ast'$. This can be seen by considering an element of the tuple $\textbf{x}$ i.e. a vertex set - a set of base vertices of $\G$. These vertices all map to the same base vertex in $\DG$, which then corresponds with the appropriate element in the tuple $h'(\textbf{x})\in\Ast'$. Therefore for every tuple $\textbf{x}\in\Xst'$ we have $h'(\textbf{x})\in\Ast'$, and hence $h':\Xst'\rightarrow\Ast'$ is a homomorphism.

The argument in the reverse direction is much simpler. Given an arbitrary homomorphism $h':\Xst'\rightarrow\Ast'$, we can easily recover the homomorphism $h:\G\rightarrow\DG$. By construction the tuples in $\Xst'$ correspond to top vertices of $\G$, while their elements are sets of bottom vertices (and new vertices). Likewise the tuples of $\Ast'$ correspond to top vertices of $\DG$, and their elements are bottom vertices. Therefore given the homomorphism $h'$ we know which top vertices of $\G$ map to which top vertices of $\DG$, and likewise for the bottom vertices. Each internal component of $\G$ is then forced into mapping onto the only satisfiable path available in $\DG$, and thus we have recovered the homomorphism $h:\G\rightarrow\DG$.

Given that we can fully determine $h'$ given $h$, and vice versa, we have successfully proven that $h$ and $h'$ are in one-to-one correspondence.

It also follows from our construction of $\rho'_w$ that the corresponding homomorphisms will have the same cost as we now show.

Consider a homomorphism $h:\G\rightarrow\DG$, it will have the following cost:
\[cost(h) = \sum_{e\in W} w_e\cdot u(h(e))\]
The corresponding homomorphism $h':\Xst'\rightarrow \Ast'$ will have the following cost:
\[cost(h')=\sum_{\substack{(A_1,\ldots,A_n) \text{ s.t.} \\ \rho'_w \text{ is defined}}} \rho'_w \cdot \rho(h'(A_1),\ldots, h'(A_n)) \]

Let $e$ be a top vertex in $\G$, then $h(e)$ is a top level vertex in $\DG$. Also let $(A_1,\ldots,A_n)$ be a tuple in $\Xst'$, then $(h'(A_1),\ldots,h'(A_n))$ is a tuple in $\Ast'$. By construction each of the top vertices of $\G$ corresponds with a tuple in $\Xst'$, and each of the top level vertices in $\DG$ corresponds with a tuple in $\Ast'$. Given the one-to-one correspondence between the homomorphisms $h$ and $h'$, if the vertex $e$ in $\G$ corresponds with the tuple $(A_1,\ldots,A_n)$ in $\Xst'$, then the vertex $h(e)$ in $\DG$ corresponds with the tuple $(h'(A_1),\ldots,h'(A_n))$ in $\Ast'$. The cost of applying the homomorphism $h$ to the vertex $e$ is $u(h(e))$. Similarly the cost of applying the homomorphism $h'$ to the corresponding tuple $(A_1,\ldots,A_n)$ is $\rho(h'(A_1),\ldots,h'(A_n))$. It follows from the definitions of $u$ and $\rho$ that we have $u(h(e)) = \rho(h'(A_1),\ldots,h'(A_n))$.

Finally we argue that costs of $h$ and $h'$ are equal. First consider the case where all of the vertices $e\in W$ map to unique top vertices in $\DG$, that is if $e_1 \neq e_2 \Rightarrow h(e_1) \neq h(e_2)$. Then $\rho'_w=w_e$ by definition, and therefore $cost(h) = cost(h')$.
Now consider the case where $h$ maps at least two vertices in $W$, say $e_1$ and $e_2$, to the same top vertex in $\DG$. This implies that both $u(h(e_1)) = \rho(h'(A_1),\ldots,h'(A_n))$ and $u(h(e_2)) = \rho(h'(A_1),\ldots,h'(A_n))$. Therefore the weight $\rho'_w$ must be the sum of the weights $w_{e_1}$ and $w_{e_2}$, in order for $cost(h) = cost(h')$, and this holds by the definition of $\rho'_w$.

If the set $W$ is empty then the cost of $h$ is undefined. In turn there would be no tuples $(A_1,\ldots,A_n)$ such that $\rho'_w$ is defined, and the cost of $h'$ would also be undefined. This reduces the problem to the feasibility problem as in~\cite{Bulin13:reduc,Bulin14:reduc}.
\end{proof}

\subsection{Dealing with Short Components}\label{pathSat}

In this section, we provide the argument justifying Stage 2 of the algorithm from the proof of Lemma~\ref{reduceDtoAunary}.

Although we consider VCSPs, when eliminating short components of an input digraph $\G$,  we first check that the components are satisfiable in $\DG$ using the method for standard CSPs as in~\cite{Bulin13:reduc}. This involves testing that components are satisfiable in some fixed family of directed paths, and then we can identify their associated costs. Thus we need the following lemma:

\begin{lemma}[$\cite{Bulin13:reduc}$]\label{pathSatCSP}
Consider the paths $\Q_S$ where $S\subseteq\{1,...,k\}$ and recall that these paths have zigzags in every position $i$ where $i\notin S$
\begin{enumerate}
	\item $\CSP(\Q_{[k]\setminus\{i\}})$ is solvable in polynomial time for any $i\in\{1\zd k\}$, even when singleton unary relations are added.   	
	\item For any $S\subseteq\{1\zd k\}$ the problem $\CSP(\Q_{S})$ is solvable in polynomial time.
\end{enumerate}
\end{lemma}

\begin{definition}
Let $\Q_{S_1}$, $\Q_{S_2}\zd \Q_{S_l}$ be paths that all have the same initial (or terminal) vertex in $\DG$. Define $\Fan$ to be the fan structure obtained when these paths are amalgamated at their shared vertex $v$, and $u|_{\Fan}$ be our unary function restricted to $\Fan$.
\end{definition}

\begin{lemma}\label{pathSatUnaryVCSP}
$\MCHom(\Fan,u|_{\Fan})$, restricted to inputs of height less than $m$, is polynomial-time solvable.
\end{lemma}

\begin{proof}
Consider an instance $H=(\H,W)$ of $\MCHom(\Fan,u|_{\Fan})$. We may assume $\H$ has a homomorphism to $\Fan$, has height strictly less than the height of $\Fan$, and is a single component. We call a homomorphism from $\H$ to $\Fan$ a solution, and $\H$ is satisfiable if it has a solution. We now consider the following cases:

\begin{enumerate}[1.]
	\item First check if $\H$ has a solution that does not involve $v$.

In this case $\H$ must be satisfiable within at least one of the paths $\Q_{S_i}$ of $\Fan$ and by applying (2) of Lemma~\ref{pathSatCSP} to every path $\Q_{S_i}$ we can identify all paths where $\H$ is feasible. It is possible that a vertex $u$ of $\H$ can be interpreted at different heights in a path $\Q_{S_i}$, and we use the unary singletons to fix a particular height for $u$ to test for a solution. If we don't find a solution fixing $u$ at that height we successively try new heights for $u$. If we find no solutions for $\H$ at any height in any path $\Q_{S_i}$ then we continue to case 2.
If there is a solution and a vertex of $\H$ maps to a top level vertex $t$ in $\Fan$ and $t\in W$ then the cost of that mapping is $u|_{\Fan}(t)$. If $\H$ has multiple solutions with non-zero cost then we choose the solution with the minimal cost, and $\H$ reduces to a single valued tuple of the objective function of $\A$ determined by $t$, with cost $u|_{\Fan}(t)$. If no vertex of $\H$ maps to a vertex in $W$ in any feasible solution then $\H$ has no influence on the optimisation problem and is ignored for the remainder of the reduction.

	\item $\H$ has a solution involving $v$.

If $v\notin W$ we follow the same procedure given in~\cite{Bulin13:reduc} as there is no optimisation to consider. That procedure is given here for completeness. First choose a vertex $h\in\H$, and check if $h$ can be interpreted at height $m$ in the same way as (1) of Lemma~\ref{pathSatCSP}. Consider the components $C_j$ of the induced subgraph obtained by removing the vertices of height $m$ from $\H$. Test every component $C_j$ for satisfaction in a path $\Q_{S_i}$, with the highest level vertices of $C_j$ constrained to be at height $m-1$ in $\Q_{S_i}$. If every component $C_j$ can be satisfied in some path $\Q_{S_i}$, then $\H$ is satisfiable in $\Fan$.
Should $\H$ not be satisfiable in $\Fan$ for a particular choice of vertex $h$ then select a new vertex $h$, and repeat until $\H$ is found to be satisfiable in $\Fan$ (as we assume $\H$ has a feasible solution).

If $v\in W$ then $\H$ reduces to a single valued tuple in the objective function of $\A$ determined by $v$, with cost $u|_{\Fan}(v)$.	
\end{enumerate}

\end{proof}

\subsection{Preservation of Algebraic Properties}\label{balancedSec}

The study of algebraic properties (e.g. polymorphisms) of constraint languages has been very useful in classifying the computational complexity of CSPs. It has been the basis of a number of important results such as the CSP dichotomy proof on 3-element domains~\cite{Bula06:threeElement} and the work of Barto and Kozik~\cite{Barto09:boundedWidth} describing constraint languages that are solvable by local consistency methods (problems of bounded width). The algebraic CSP dichotomy conjecture~\cite{Bula05:algdichotomy} predicts, in terms of polymorphisms, where the split between polynomial time and $\NP$-complete problems occurs.

The important properties of polymorphisms are usually given by {\em identities}, i.e. equalities of terms that hold for all choices of the variables
involved in them. Here are some of the important types of operations:
\begin{itemize}
\item An operation $f$ is {\em idempotent} if it satisfies the identity $f(x,\ldots,x)=x$.
\item A $k$-ary ($k\ge 2$) operation $f$ is {\em weak near unanimity (WNU)} if it is idempotent and satisfies the identities
$f(y,x,\dots,x,x)=f(x,y,\dots,x,x)=\cdots=f(x,x,\dots,x,y).$
\item A $k$-ary ($k\ge 2$) operation $f$ is {\em cyclic} if $f(x_1,x_2,\dots,x_k)=f(x_2,\dots,x_k,x_1)$.
\item A $k$-ary ($k\ge 2$) operation $f$ is {\em symmetric} if $f(x_1,\dots,x_k)=f(x_{\pi(1)},\dots,x_{\pi(k)})$ for each permutation $\pi$ on $\{1,\dots,k\}$.
\end{itemize}
For example, the algebraic dichotomy conjecture can be re-stated as follows~\cite{Bula05:algdichotomy,Maroti08:existence}: for a core structure $\Ast$, $\CSP(\Ast)$ is tractable if
$\Ast$ has a WNU polymorphism of some arity, and $\NP$-complete otherwise. For a core $\Ast$, the problems $\CSP(\Ast)$ has bounded width if and only if $\Ast$ has WNU polymorphisms of almost all arities~\cite{Barto09:boundedWidth}. For a core weighted structure $w\Ast$ such that each weighted relation in $w\Ast$ is defined on all tuples of the corresponding arity, $\VCSP(w\Ast)$ is solvable in polynomial time if
$w\Ast$ has symmetric weighted polymorphisms of all arities~\cite{Thap12:dichotomy}, and it is $\NP$-hard otherwise.

It is well known and easy to see that a (weighted or unweighted) structure is a rigid core if and only if all its polymorphisms are idempotent.

An operational signature is a set of operation symbols with arities assigned to them. An identity is an expression $t_1=t_2$ where $t_1$ and $t_2$ are terms in this signature. An identity $t_1=t_2$ is said to be {\em linear} if both $t_1$ and $t_2$ involve at most one occurrence of an operation symbol, and {\em balanced} if the variables occuring in $t_1$ and $t_2$ are the same (e.g. $f(x,x,y)=g(y,x,x)$). (This notion is not related to balanced digraphs). A set $\Sigma$ of identities  is linear if it only contains linear identities, idempotent if for each operation symbol, $f$, the identity $f(x,x,...,x)=x$ is in $\Sigma$ and balanced if all of the identities in $\Sigma$ are balanced.
Note the identities defining WNU, symmetric and cyclic operations above are linear and balanced.

Recall the structure $w\Ast_0$ from Lemma~\ref{expressiblelem2}. Since $w\Ast$ is expressible in $w\Ast_0$, every weighted polymorphism
of $w\Ast_0$, when restricted to $D$, is a weighted polymorphism of $w\Ast$ (see~\cite{Cohen12:weightedPol}). Hence the presence of a weighted polymorphism $\omega:C\rightarrow \Q_+$ such that the operations in $C$ satisfy some set of identities carries over from $w\Ast_0$ to $w\Ast$.
We show that, for linear balanced sets of identities, the converse is also true.

First we must introduce some facts about connected components of powers of $\DG$. Let $\DG^k$ be the direct $k$th power of the digraph $\DG$, i.e. its vertices are the $k$-tuples of elements of $\DG$, and $(\mathbf{c},\mathbf{d})$ is an edge in $\DG^k$ if and only if, for all $1\le i\le k$, $(c_i,d_i)$ is an edge in $\DG$. Consider the diagonal of $\DG$, i.e. the set $\{(c,\ldots,c)\mid c\in V^\DG\}$. Clearly, the diagonal is contained in one (weakly) connected component of $\DG^k$, denote it by $\Delta_k$. We will need some properties of $\DG^k$ proven in~\cite{Bulin14:reduc}.

\begin{lemma}[\cite{Bulin14:reduc}]
We have both $D^k\subseteq \Delta_k$ and $R^k\subseteq \Delta_k$.
\end{lemma}

\begin{lemma}[\cite{Bulin14:reduc}]
Assume that a connected component $\Delta'$ of $\DG^k$ contains a tuple $\mathbf{c}=(c_1,\ldots,c_k)$ such that
$lvl(c_1)=\ldots =lvl(c_k)$. Then every $\mathbf{d}=(d_1,\ldots,d_k)$ satisfies $lvl(d_1)=\ldots =lvl(d_k)$ and also
either $\Delta'=\Delta_k$ or $\Delta'$ is one-element.
\end{lemma}

As in~\cite{Bulin13:reduc,Bulin14:reduc} we define a linear ordering on the vertices of the digraph $\DG$. For every $e=(a,\mathbf{r})\in D\times R$, denote the path $\Q_{\{i:a=r_i\}}$ in $\DG$ by $\P_e$. Also, write $\P_{e,l}$ to mean $\Q_{S,l}$ where $\P_e=\Q_S$. First, fix a linear ordering $\preceq_1$ on $D$ and extend it to any linear ordering $\preceq$ of $E=D\times R$ such that if $(d_1,\mathbf{t}_1)\preceq (d_2,\mathbf{t}_2) \preceq (d_1,\mathbf{t}_3)$ then $d_1=d_2$.
Now define the mapping $\epsilon:V^\DG\rightarrow E$ by setting $\epsilon(x)$ to be the $\preceq$-minimal $e\in E$ such that $x\in\P_e$. Finally we define the linear order $\sqsubseteq$ on the vertices of the digraph $\DG$, where $x\sqsubset y$ if either:\\
$lvl(x) < lvl(y)$, or \\
$lvl(x) = lvl(y)$ and $\epsilon(x)\prec\epsilon(y)$, or \\
$lvl(x) = lvl(y)$, $\epsilon(x)=\epsilon(y)$, and $x$ is closer to $\iota\P_{\epsilon(x)}$ than $y$.

\begin{lemma}[\cite{Bulin14:reduc}]
Let $K$ and $L$ be subsets of $V^\DG$ such that $L\not\subseteq R$ and
\begin{itemize}
\item for every $x\in K$ there is $y'\in L$ such that $x\rightarrow y'$ is an edge in $\DG$, and
\item for every $y\in L$ there is $x'\in K$ such that $x'\rightarrow y$ is an edge in $\DG$.
\end{itemize}
If $c$ and $d$ are the $\sqsubset$-minimal elements of $K$ and $L$, respectively, then $c\rightarrow d$ is an edge in $\DG$.
\end{lemma}

Now we introduce the main theorem of this section. 

\begin{theorem}\label{thm:algebra}
Let $w\Ast$ and $(\DG,u)$ be as in Theorem~\ref{thm:main}. If $w\Ast$ has a $k$-ary weighted polymorphism $\omega:C\rightarrow \Q_+$
such that operations in $C$ satisfy a linear balanced set $\Sigma$ of identities then $(\DG,u)$ also has a $k$-ary weighted polymorphism
$\omega_0:C_0\rightarrow \Q_+$ such that there is a bijection between $C$ and $C_0$, and the operations in $C_0$ satisfy $\Sigma$.
In particular, if $\omega$ is such that some operation in $C$ (or all non-projection operations in $C$) is WNU then the same holds for $\omega_0$.
Similarly, if $\omega$ is such that some operation in $C$ (or all non-projection operations in $C$) is cyclic or symmetric then the same holds for $\omega_0$.
\end{theorem}

\begin{proof} It is shown in~\cite{Bulin13:reduc,Bulin14:reduc} how polymorphisms of $\rho$ can be transformed (in fact, extended) to polymorphisms of $\DG$
in such a way that any set of linear balanced identities carries over. The transformation there is designed to preserve not only balanced
identities, and, for our purposes, we can use a simplified version of it. Let $C_0$ be obtained from $C$  by
applying this (simplified) transformation to all non-projection operations in $C$ and extending projection operations in $C$ so that they stay projection operations. Obtain $\omega_0$ from $\omega$ by using this bijection between $C_0$ and $C$, i.e. keep the weights of operations the same.
Then $\omega_0$ will be a weighted polymorphism of $\DG$. Indeed, since $\DG$ is 0-weighted, the last condition in the definition of a weighted polymorphism will be trivially satisfied, while the other conditions trivially carry over.
Hence, it only remains to ensure that $\omega_0$ is a weighted polymorphism of $u$. For this, we extend the operations $f$ from $C$ to operations $f_0$ on $V^{\DG}$
in such a way that, for any $v_1,\ldots,v_k\in V^{\DG}$, we have $f_0(v_1,\ldots,v_k)\in R$ only if $v_1,\ldots,v_k\in R$.
With this condition, the fact that $\omega_0$ is a weighted polymorphism of $u$ follows from the fact
that $\omega$ is a weighted polymorphism of $\rho$, as we show in the rest of the proof.

Let $\Sigma$ be a set of linear balanced identities in operations symbols $\{f_\lambda \mid \lambda\in \Lambda\}$
such that, interpreting each $f_\lambda$ as a specific operation $f_\lambda^\Ast\in C$, the operations
$\{f_\lambda^\Ast \mid \lambda\in \Lambda\}$ satisfy $\Sigma$. We can without loss of generality assume that
$\{f_\lambda^\Ast \mid \lambda\in \Lambda\}$ is the set of all non-projection operations in $C$.

We will extend each projection operation in $C$ to the corresponding projection on $V^\DG$ and each non-projection operation $f_\lambda^\Ast\in C$ to a polymorphism $f_\lambda^\DG$ of $\DG$
in such a way that $\{f_\lambda^\DG \mid \lambda\in \Lambda\}$ will also satisfy $\Sigma$.
The construction will also ensure that $\omega_0$ obtained from $\omega$ as described above is indeed a weighted polymorphism of $u$.

As in~\cite{Bulin13:reduc,Bulin14:reduc}, let the digraph $\Z$ be the zigzag with vertices labelled 00, 01, 10 and 11, such that we describe the oriented path 00 $\rightarrow$ 01 $\leftarrow$ 10 $\rightarrow$ 11. Given a vertex pair $\{x,y\}$ in the zigzag, define the operation $\wedge$
such that $x\wedge y$ is the vertex closer to 00. For each $\lambda\in \Lambda$, let $f_\lambda^{\Z}(x_1,...,x_k)=\bigwedge^{k}_{i=1}x_i$
where $k$ is the arity of $f_\lambda$. It is clear that the set $\{f_\lambda^\Z \mid \lambda\in \Lambda\}$ satisfies any balanced set of identities.

Now we define polymorphisms $\{f_{\lambda}^{\DG}\mid\lambda\in\Lambda\}$. Fix $\lambda\in\Lambda$, assume that $f_\lambda$ is a $k$-ary non-projection operation and let $\textbf{c}\in (V^\DG)^k$. If $\mathbf{c}\in R^k$ then $(f_{\lambda}^{\Ast})^{(k)}(\textbf{c})$ will denote the element of $R$ obtained by applying
$f_\lambda^\Ast$ to the tuples $c_1,\ldots,c_k\in R$ component-wise. Note that $(f_{\lambda}^{\Ast})^{(k)}(\textbf{c})\in R$ because
$f_\lambda^\Ast$ is a polymorphism of $R$.
Similarly, we can apply $f_\lambda^\Ast$ to elements $e_1,\ldots,e_k \in D\times R$ and obtained again an element $(f_{\lambda}^{\A})^{k+1}(e_1,...,e_k)$ from $D\times R$.

Construct $f_{\lambda}^{\DG}$ as follows:\\
\\
$\textbf{Case 1.}$ $\textbf{c}\in D^k \cup R^k$.\\
1a. If $\textbf{c}\in D^k$, we define $f_{\lambda}^{\DG}(\textbf{c}) = f_{\lambda}^{\Ast}(\textbf{c})$. \\
1b. If $\textbf{c}\in R^k$, we define $f_{\lambda}^{\DG}(\textbf{c}) = (f_{\lambda}^{\Ast})^{(k)}(\textbf{c})$. \\
\\
$\textbf{Case 2.}$ $\textbf{c}\in\Delta_k \backslash (D^k \cup R^k)$.\\
Let $e_i=\epsilon(c_i)$ for $1\le i\le k$ and $e=(f_{\lambda}^{\Ast})^{k+1}(e_1,...,e_k)$. Let $1\le l\le k$ be minimal such that $c_i\in\P_{e_i,l}$ for all $1\le i\le k$. \\
2a. If $\P_{e,l}$ is a single edge, then we define $f_{\lambda}^{\DG}(\textbf{c})$ to be the vertex from $\P_{e,l}$ having the same level as all the $c_i$'s. \\
If $\P_{e,l}$ is a zigzag then at least one of the $\P_{e_i,l}$'s is a zigzag as well. For every $1\le i\le k$ such that $\P_{e_i,l}$ is a zigzag let $\Phi_i:\P_{e_i,l}\rightarrow\Z$ be the unique isomorphism. Let $\Phi$ denote the isomorphism from $\P_{e,l}$ to $\Z$.  \\
2b. If all of the $\P_{e_i,l}$'s are zigzags, then $f_{\lambda}^{\DG}(\textbf{c})=\Phi^{-1}(f_{\lambda}^{\Z}(\Phi_1(c_1),...,\Phi_m(c_k)))$.  \\
2c. Else, we define $f_{\lambda}^{\DG}(\textbf{c})$ to be the $\sqsubseteq$-minimal element from the set $\{\Phi^{-1}(\Phi_i(c_i))|\P_{e_i,l}$ is a zigzag$\}$  \\
\\
$\textbf{Case 3.}$ $\textbf{c}\notin\Delta_k$.\\
Define $f_{\lambda}^{\DG}(\textbf{c})$ to be the $\sqsubseteq$-minimal element from the set $\{c_1,...,c_k\}$.
\\
\\
The definition of $f_\lambda^\DG$ in~\cite{Bulin14:reduc} is similar, but Case 3 there is split into three subcases (3a)-(3c), which is unnecessary for our purposes, as we use their (3c) throughout our Case 3.
The proof that $f_\lambda^\DG$ is a polymorphism of $\DG$ is a subset of the proof of Claim 5.7 in~\cite{Bulin14:reduc}.

It remains to show that $\omega_0$ is a weighted polymorphism of $u$, i.e. $\omega_0$ and $u$ satisfy the third condition
in the definition of a weighted polymorphism. When applied to $u$, this condition says that, for any $x_1,\ldots,x_k\in V^{\DG}$,
we have $\sum_{f\in C_0}{\omega_0(f)\cdot u(f(x_1,\ldots,x_k))}\le 0$.
Recall that, by definition, $u(x)=0$ for all $x\in V^{\DG}\backslash R$.
By inspecting our definition of $f_\lambda^\DG$, it is clear that if $f_\lambda^\DG(x_1,\ldots,x_k)\in R$ then
$x_1,\ldots,x_k\in R$. Thus, if not all $x_1,\ldots,x_k$ are in $R$, the only possible non-0 terms $u(f(x_1,\ldots,x_k))$ in the sum correspond to projections (whose weights are non-positive by definition), and hence the whole sum is non-positive. On the other hand, if all $x_1,\ldots,x_k$ are in $R$ then $f(x_1,\ldots,x_k)\in R$ and so $u(f(x_1,\ldots,x_k))=\rho(f(x_1,\ldots,x_k))$. In this case, the inequality holds because the inequality $\sum_{f\in C}{\omega(f)\cdot \rho(f(x_1,\ldots,x_k))}\le 0$ holds for $\omega$.

This finishes the proof of Theorem~\ref{thm:algebra}.
\end{proof}


\bibliographystyle{plain}

\begin{thebibliography}{10}

\bibitem{Barto09:boundedWidth}
L.~Barto, and M.~Kozik.
\newblock Constraint {S}atisfaction {P}roblems {S}olvable by {L}ocal
  {C}onsistency {M}ethods.
\newblock \emph{Journal of the ACM}, 61(1), 2014.
\newblock Article No. 3.

\bibitem{Bula06:threeElement}
A.~Bulatov.
\newblock A Dichotomy Theorem for Constraint Satisfaction Problems on a 3-Element Set.
\newblock {\em Journal of the ACM}, 53(1), 66--120, 2006.

\bibitem{Bula05:algdichotomy}
A.~Bulatov, P.~Jeavons, and A.~Krokhin.
\newblock Classifying the Complexity of Constraints Using Finite Algebras.
\newblock {\em SIAM Journal on Computing}, 34(3), 720--742, 2005.

\bibitem{Bulin13:reduc}
J.~Bulin, D.~Delic, M.~Jackson, and T.~Niven.
\newblock On the reduction of the CSP dichotomy conjecture to digraphs.
\newblock In {\em CP'13, volume 8124 of LNCS}, 184--199, 2013.

\bibitem{Bulin14:reduc}
J.~Bulin, D.~Delic, M.~Jackson, and T.~Niven.
\newblock A finer reduction of constraint problems to digraphs.
\newblock Technical report, arXiv:1406.6413, 2014.

\bibitem{Cohen12:weightedPol}
D.~Cohen, M.~Cooper, P.~Creed, P.~Jeavons, and S.~\v{Z}ivn\'y.
\newblock An algebraic theory of complexity for discrete optimisation.
\newblock In {\em SIAM Journal on Computing}, 42(5), 1915--1939, 2013.

\bibitem{Cohen06:algebraic}
D.~Cohen, M.~Cooper, and P.~Jeavons.
\newblock An algebraic characterisation of complexity for valued constraints.
\newblock In {\em CP'06, volume 4204 of LNCS}, 107--121, 2006.


\bibitem{Cohen06:soft}
D.~Cohen, M.~Cooper, P.~Jeavons, and A.~Krokhin.
\newblock The complexity of soft constraint satisfaction.
\newblock {\em Artificial Intelligence}, 170(11):983--1016, 2006.

\bibitem{FederVardi98:cs}
T.~Feder, and M.~Vardi.
\newblock The Computational Structure of Monotone Monadic SNP and Constraint Satisfaction: A Study Through Datalog and Group Theory.
\newblock {\em SIAM Journal on Computing}, 28:57-104, 1998.

\bibitem{Gutin:mincostdichotomy}
G.~Gutin, P.~Hell, A.~Rafiey, and A.~Yeo.
\newblock A dichotomy for minimum cost graph homomorphisms.
\newblock \emph{European Journal of Combinatorics}, 29(4):
  900--911, 2008.

\bibitem{Hell04:book}
P.~Hell and J.~Ne\v{s}et\v{r}il.
\newblock \emph{Graphs and {H}omomorphisms}.
\newblock Oxford University Press, 2004.

\bibitem{Hell12:sidma}
P.~Hell and A.~Rafiey.
\newblock {The Dichotomy of Minimum Cost Homomorphism Problems for Digraphs}.
\newblock \emph{SIAM Journal on Discrete Mathematics}, 26(4): 1597--1608, 2012.

\bibitem{Huber13:skew}
A.~Huber, A.~Krokhin, and R.~Powell.
\newblock Skew Bisubmodularity and Valued CSPs.
\newblock \emph{SIAM Journal on Comptuing}, 43(3): 1064–-1084, 2014.

\bibitem{Jeavons14:survey}
P.~Jeavons, A.~Krokhin, and S.~\v{Z}ivn\'y.
\newblock The complexity of valued constraint satisfaction.
\newblock {\em Bulletin of the BEATCS}, 113:21--55, 2014.

\bibitem{Kolmogorov15:complexity}
V.~Kolmogorov, A.~Krokhin, and M.~Rol\'inek.
\newblock The Complexity of General-Valued CSPs.
\newblock Technical report, arXiv:1502.07327, 2015.

\bibitem{Maroti08:existence}
M.~Mar\'oti and R.~McKenzie.
\newblock Existence theorems for weakly symmetric operations.
\newblock \emph{Algebra Universalis}, 59(3-4):463--489,
  2008.

\bibitem{Och12:rigidCore}
M.~Kozik and J.~Ochremiak.
\newblock Algebraic Properties of Valued Constraint Satisfaction Problems.
\newblock Technical report, arXiv:1403.0476, 2015.

\bibitem{Sch78:comp}
T.~Schaefer.
\newblock The Complexity of Satisfiability Problems.
\newblock In {\em STOC'78}, 216--226, 1978.

\bibitem{Tak10:MinCostHom}
R.~Takhanov.
\newblock A Dichotomy Theorem for the General Minimum Cost Homomorphism Problem.
\newblock In {\em STACS'10}, 657--668, 2010.

\bibitem{Tak12:extMinCostHom}
R.~Takhanov.
\newblock Extensions of the Minimum Cost Homomorphism Problem.
\newblock In {\em COCOON'10}, 328--337, 2010.

\bibitem{Thap12:linearProg}
J.~Thapper and S.~\v{Z}ivn\'y.
\newblock The Power of Linear Programming for Valued CSPs.
\newblock In {\em FOCS'12}, 669-678, 2012.

\bibitem{Thap12:dichotomy}
J.~Thapper and S.~\v{Z}ivn\'y.
\newblock The Complexity of Finite-Valued CSPs.
\newblock In {\em STOC'13}, 695--704, 2013.

\bibitem{Uppman13:icalp}
H.~Uppman.
\newblock The {C}omplexity of {T}hree-{E}lement {M}in-{S}ol and {C}onservative
  {M}in-{C}ost-{H}om.
\newblock In \emph{ICALP'13}, 804--815, 2013.

\bibitem{Uppman14:mincosthom}
H.~Uppman.
\newblock Computational Complexity of the Extended Minimum Cost Homomorphism Problem on Three-Element Domains
\newblock In {\em STACS 2014}, 651-662, 2014.

\end{thebibliography}

\end{document}